\documentclass[aps,prx,twocolumn,english,showpacs,superscriptaddress,amssymb,amsfonts,nofootinbib]{revtex4-2}
\usepackage{graphicx}
\usepackage{hyperref}
\usepackage{amsmath}
\usepackage{tikz}
\usetikzlibrary{positioning}
\usepackage{soul,xcolor}
\usepackage[export]{adjustbox}
\usepackage{placeins}

\newcommand{\Rom}[1]{\uppercase\expandafter{\romannumeral #1\relax}}
\usepackage{amsthm}

\newtheorem*{assumption*}{\assumptionnumber}
\providecommand{\assumptionnumber}{}


\theoremstyle{definition}
\newtheorem{definition}{Definition}[section]

\newtheorem*{remark}{Remark}

\theoremstyle{theorem}
\newtheorem{theorem}{Theorem}[section] 

\theoremstyle{corollary}

\theoremstyle{lemma} 
\newtheorem{lemma}[theorem]{Lemma}

\theoremstyle{Proposition}

\theoremstyle{definition}
\newtheorem{exmp}{Example}[section]

\usepackage{amssymb}

\newcommand{\calH}{{\mathcal H}}

\newcommand{\Tr}{{\rm Tr}}

\newcommand{\Rk}{{\rm Rank}_{[2]}}

\definecolor{XDblue}{RGB}{25,25,125}
\definecolor{BSorange}{RGB}{140,50,0}
\definecolor{YMLgreen}{RGB}{50,100,50}

\begin{document}
	\title{Entanglement negativity at the critical point of measurement-driven transition}
	\author{Bowen Shi}
	\affiliation{Department of Physics, The Ohio State University, Columbus, OH 43210, USA}
	\affiliation{Department of Physics, University of California at San Diego, La Jolla, CA 92093, USA}
	\author{Xin Dai}
	\affiliation{Department of Physics, The Ohio State University, Columbus, OH 43210, USA}
	\affiliation{Computational Science Initiative, Brookhaven National Laboratory, Upton, NY 11973, USA}
	\author{Yuan-Ming Lu}
	\affiliation{Department of Physics, The Ohio State University, Columbus, OH 43210, USA}
	
	\date{\today}
	
	\begin{abstract}
	We study the entanglement behavior of a random unitary circuit punctuated by projective measurements at the measurement-driven phase transition in one spatial dimension. We numerically study the logarithmic entanglement negativity of two disjoint intervals and find that it scales as a power of the cross-ratio. We investigate two systems: (1) Clifford circuits with projective measurements, and (2) Haar random local unitary circuit with projective measurements.
	Remarkably, we identify a power-law behavior of entanglement negativity at the critical point. Previous results of entanglement entropy and mutual information point to an emergent conformal invariance of the measurement-driven transition. Our result suggests that the critical behavior of the measurement-driven transition is distinct from the ground state behavior of any \emph{unitary} conformal field theory.
    \end{abstract}

	\maketitle

	\section{Introduction and main results}

Entanglement is a central concept in quantum physics, and it enables some of the most important applications of quantum mechanics such as quantum teleportation and quantum computation. For bipartite pure states, the entanglement entropy measures the entanglement of subsystem $A$ with its complement. In recent years, the entanglement entropy has proven to be an insightful tool to characterize quantum many-body phases\cite{2006PhRvL..96k0404K,2006PhRvL..96k0405L,2006PhRvL..96r1602R,2008RvMP...80..517A,zeng2019quantum}.

The entanglement entropy, however, is not a complete characterization of quantum entanglement\cite{2006PhDT........59E,2009RvMP...81..865H,2017arXiv170102187D}. While entanglement entropy and associated mutual information provide a good characterization of entanglement in a pure quantum state, they measure the total amount of correlation rather than the quantum entanglement between two regions in a mixed quantum state\cite{2005PhRvA..72c2317G,2008PhRvL.100g0502W}. How to quantify the quantum entanglement for bipartite mixed states (or equivalently tripartite pure states)? One such calculable measure is known as the entanglement negativity\cite{2002PhRvA..65c2314V,2006PhDT........59E}. Recently entanglement negativity as well as other mixed states entanglement measures have been applied to enrich our understanding of various quantum many-body systems\cite{2012PhRvL.109m0502C,2013JSMTE..02..008C,marvian2017symmetry,2013PhRvA..88d2318L,2016PhRvB..94g5124W,2017PhRvB..95p5101S,2019arXiv191204293L,2021PhRvL.126l0501Z}.

Recently, there are significant interests in the entanglement properties of a 1D dynamical quantum system, driven by
random local unitary circuits and random local measurements with a certain rate $0<p<1$ \cite{2018PhRvB..98t5136L,2019PhRvX...9c1009S,2019PhRvB..99v4307C,2019PhRvB.100m4306L,2020PhRvL.125c0505C,2020PhRvB.101j4302J,2020PhRvB.101j4301B,2019PhRvB.100f4204S,2019arXiv190505195G,2019arXiv191000020G,2020PhRvB.101f0301Z,2020PhRvB.101w5104Z,2020arXiv200212385F,2020arXiv200312721L}. In such random unitary circuits with measurements, the unitary evolution tends to increase the entanglement entropy~\cite{2017PhRvX...7c1016N} whereas the local measurements tend to disentangle the system. This leads to a volume-law entangled phase at a low rate of measurement $p<p_c$ and an area-law entangled phase at a high rate of measurement $p>p_c$~\cite{2018PhRvB..98t5136L,2019PhRvX...9c1009S,2019PhRvB..99v4307C}. After that, a number of works have studied various properties of the two phases and the critical point that separates them~\cite{2019PhRvB.100m4306L,2020PhRvL.125c0505C,2020PhRvB.101j4302J,2020PhRvB.101j4301B,2019PhRvB.100f4204S,2019arXiv190505195G,2019arXiv191000020G,2020PhRvB.101f0301Z,2020PhRvB.101w5104Z,2020arXiv200212385F,2020arXiv200312721L}. Most of the previous works focused on the entanglement entropy and mutual information between different regions in the system. 

In this work, our goal is to characterize the quantum entanglement in the 1D random unitary circuits with measurements by numerically studying the entanglement negativity between two disjoint intervals. In particular, we focus on the critical behavior of the entanglement negativity at the measurement-driven phase transition at $p=p_c$. We investigate both the Clifford random unitary circuit and the Haar random unitary circuit, punctuated by single-site projective measurements at a rate of $p$. We numerically identify a simple scaling law of the entanglement negativity right at the critical point. For two small intervals separated by a distance $r$, the negativity scales as $r^{-2 \Lambda}$, where $\Lambda \approx 3$ for both Clifford and Haar circuit models (within the error bar). Alternatively, the scaling law can be written as $\eta^{\Lambda}$ for $\eta \ll 1$, where $\eta$ is the cross-ratio defined by (\ref{eq:cross_ratio}) for the pair of intervals.

It is insightful to compare the power-law scaling of negativity studied in this work to previously obtained power-law scaling of mutual information at the same measurement-driven phase transition~\cite{2019PhRvB.100m4306L}. The mutual information (\ref{mutual information}) at the critical point has been shown to scale in a power-law fashion $I_{A,B}\propto\eta^{\Delta}$ with the cross-ratio $\eta$ for $\eta\ll1$, where $\Delta \approx 2$. In other words, the mutual information for two small disjoint intervals decays as $r^{-2 \Delta}$ with respect to the distance $r$ between the two intervals; see also Ref.~\cite{2019PhRvX...9c1009S}. $\Delta<\Lambda$ indicates that the entanglement negativity decays faster than the mutual information as we increase the distance between the two intervals. For a mixed quantum state, the bipartite mutual information $I_{A,B}$ measures the total correlation between the two regions\cite{2005PhRvA..72c2317G,2008PhRvL.100g0502W}, whereas the negativity as an entanglement monotone measures the quantum entanglement between them\cite{1998PhRvA..58..883Z,2002PhRvA..65c2314V,2006PhDT........59E}. Therefore our results indicate as the distance between two regions increases, that the quantum entanglement decays faster than the total correlation between them.

The most important implication of the power-law negativity scaling is on the nature of the measurement-driven phase transition. Many previous results, such as the scalings of entanglement entropy and mutual information, point to an emergent conformal invariance at this dynamical critical point\cite{2019PhRvB.100m4306L,2020arXiv200312721L,2020PhRvB.101j4302J}. While the precise nature of this emergent conformal symmetry is under debate, it remained possible that the entanglement property of the late time quantum states being similar to that in the continuous phase transitions in equilibrium, at least phenomenologically. In other words, the long-wavelength physics of the measurement-driven entanglement transition may be captured by that of the ground state of some conformal field theory (CFT). While the power-law mutual information is consistent with a unitary CFT\cite{2009JSMTE..11..001C,2011JSMTE..01..021C,2011JSMTE..01..021C,2010PhRvD..82l6010H}, it is known that the entanglement negativity of two disjoint intervals must decay faster than any power law w.r.t. a small cross-ratio $\eta\ll1$ in a unitary CFT ground state\cite{2012PhRvL.109m0502C,2013JSMTE..02..008C,2015JSMTE..06..021D}. Therefore our observation on a power-law scaling of entanglement negativity shows strong evidence that the critical behavior at the measurement-driven entanglement transition is distinct from the ground state behavior of any unitary CFT. This puts constraints on the nature that conformal symmetry enters this story.
This analysis can also provide a constraint on the possible non-unitary CFT description of the critical point\cite{2020PhRvB.101j4302J,2020arXiv200312721L}.

This paper is organized as follows.
In Sec. \ref{sec:setup}, we setup the Clifford and Haar random circuit models with projective measurements and introduce the correlation and entanglement measures (mutual information and entanglement negativity) used in this work. In Sec. \ref{sec:results}, we present our numerical results on the scaling behavior of mutual information and entanglement negativity at the measurement-driven phase transition and discuss their implications. In Sec. \ref{sec:summary}, we summarize the main conclusions of our work and outlook for the future.

\section{Setup and background}\label{sec:setup}
	
	\subsection{The models}

	We consider a 1D system consisting of $L$ qubits arranged on a ring with a periodic boundary condition (PBC), as illustrated in Fig.~\ref{fig:lattice}. We assume $L$ is even throughout the paper. The PBC 
	makes it easy to study the dependence of entanglement on the cross-ratio. We study two types of models: (1) the random Clifford circuit with single-site projective measurements, (2) the Haar random unitary circuit model with single-site projective measurements. We fix the initial state to be $\vert \psi(0)\rangle = \otimes_{i=1}^L \vert 0\rangle_i$, i.e. the product state with all spin up in the Pauli $Z$-basis. The details of the two models are described below.

	\begin{figure}[h]
		\centering
		\begin{tikzpicture}
		\begin{scope}[xshift=2.6 cm, yshift=5.5 cm, scale=0.85]

		\draw[line width=1pt] (0,0) circle (1.6);

		\fill[green!20!white] (0:1.6) circle (0.1);
		\draw[line width=1pt] (0:1.6) circle (0.1);
		
		\fill[green!20!white] (30:1.6) circle (0.1);
		\draw[line width=1pt] (30:1.6) circle (0.1);
		
		\fill[green!20!white] (60:1.6) circle (0.1);
		\draw[line width=1pt] (60:1.6) circle (0.1);
		
		\fill[green!20!white] (90:1.6) circle (0.1);
		\draw[line width=1pt] (90:1.6) circle (0.1);
		
		\fill[green!20!white] (120:1.6) circle (0.1);
		\draw[line width=1pt] (120:1.6) circle (0.1);
		
		\fill[green!20!white] (150:1.6) circle (0.1);
		\draw[line width=1pt] (150:1.6) circle (0.1);
		
		\fill[green!20!white] (180:1.6) circle (0.1);
		\draw[line width=1pt] (180:1.6) circle (0.1);
		
		\fill[green!20!white] (210:1.6) circle (0.1);
		\draw[line width=1pt] (210:1.6) circle (0.1);
		
		\fill[green!20!white] (240:1.6) circle (0.1);
		\draw[line width=1pt] (240:1.6) circle (0.1);
		
		\fill[green!20!white] (270:1.6) circle (0.1);
		\draw[line width=1pt] (270:1.6) circle (0.1);
		
		\fill[green!20!white] (300:1.6) circle (0.1);
		\draw[line width=1pt] (300:1.6) circle (0.1);
		
		\fill[green!20!white] (330:1.6) circle (0.1);
		\draw[line width=1pt] (330:1.6) circle (0.1);
		\end{scope}
		\end{tikzpicture}
		\caption{The 1D system consists of $L$ qubits (sites) arranged on a circle with the periodic boundary condition.}
		\label{fig:lattice}
	\end{figure}
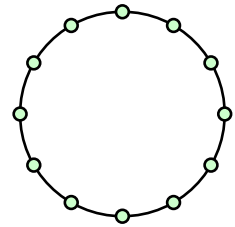

	\subsubsection{Random Clifford circuits with measurements}
	The dynamic of the system is given by random Clifford unitary circuits on pairs of nearby qubits and projective measurements of single-site Pauli $Z$ operators independently with probability $p$ for each site at any time step. See Fig.~\ref{fig:circuit} for an illustration of the bricklayer arrangement of the 2-qubit unitary operators and the position of measurements in discretized spacetime.

	The advantage of the Clifford circuit model is that it allows an efficient simulation on classical computers~\cite{1998quant.ph..7006G}. We have simulated the Clifford circuit model for system sizes up to size $L=768$. Similar to previous studies, our physical quantities are obtained by taking an ensemble average over late time (pure) quantum states. We shall refer to this average as \emph{late time average}.

	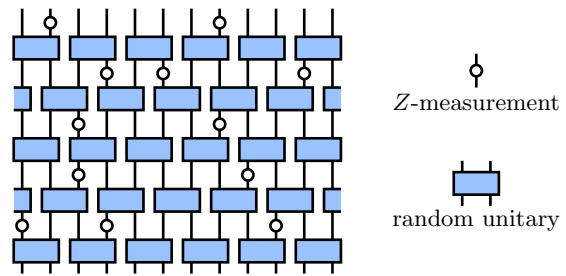
\begin{figure}[h]
		\centering
		\begin{tikzpicture}		
		\begin{scope}[xshift=0,scale=0.75]
		\draw[line width=1 pt] (-0.25, -0.4) -- (-0.25, 4.3);
		\draw[line width=1 pt] (-0.25+0.5, -0.4) -- (-0.25+0.5, 4.3);
		\draw[line width=1 pt] (-0.25+1, -0.4) -- (-0.25+1, 4.3);
		\draw[line width=1 pt] (-0.25+1.5, -0.4) -- (-0.25+1.5, 4.3);
		\draw[line width=1 pt] (-0.25+2, -0.4) -- (-0.25+2, 4.3);
		\draw[line width=1 pt] (-0.25+2.5, -0.4) -- (-0.25+2.5, 4.3);
		\draw[line width=1 pt] (-0.25+3, -0.4) -- (-0.25+3, 4.3);
		\draw[line width=1 pt] (-0.25+3.5, -0.4) -- (-0.25+3.5, 4.3);
		\draw[line width=1 pt] (-0.25+4, -0.4) -- (-0.25+4, 4.3);
		\draw[line width=1 pt] (-0.25+4.5, -0.4) -- (-0.25+4.5, 4.3);
		\draw[line width=1 pt] (-0.25+5, -0.4) -- (-0.25+5, 4.3);
		\draw[line width=1 pt] (-0.25+5.5, -0.4) -- (-0.25+5.5, 4.3);
		
		\fill[blue!60!cyan!40!white] (-0.4,-0.2) rectangle (0.4,0.2);
		\fill[blue!60!cyan!40!white] (-0.4,-0.2+1.8) rectangle (0.4,0.2+1.8);
		\fill[blue!60!cyan!40!white] (-0.4,-0.2+3.6) rectangle (0.4,0.2+3.6);
		
		\fill[blue!60!cyan!40!white] (-0.4+0.5,-0.2+0.9) rectangle (0.4+0.5,0.2+0.9);
		\fill[blue!60!cyan!40!white] (-0.4+0.5,-0.2+2.7) rectangle (0.4+0.5,0.2+2.7);
		
		\fill[blue!60!cyan!40!white] (-0.4+1,-0.2) rectangle (0.4+1,0.2);
		\fill[blue!60!cyan!40!white] (-0.4+1,-0.2+1.8) rectangle (0.4+1,0.2+1.8);
		\fill[blue!60!cyan!40!white] (-0.4+1,-0.2+3.6) rectangle (0.4+1,0.2+3.6);
		
		\fill[blue!60!cyan!40!white] (-0.4+1.5,-0.2+0.9) rectangle (0.4+1.5,0.2+0.9);
		\fill[blue!60!cyan!40!white] (-0.4+1.5,-0.2+2.7) rectangle (0.4+1.5,0.2+2.7);
		
		\fill[blue!60!cyan!40!white] (-0.4+2,-0.2) rectangle (0.4+2,0.2);
		\fill[blue!60!cyan!40!white] (-0.4+2,-0.2+1.8) rectangle (0.4+2,0.2+1.8);
		\fill[blue!60!cyan!40!white] (-0.4+2,-0.2+3.6) rectangle (0.4+2,0.2+3.6);
		
		\fill[blue!60!cyan!40!white] (-0.4+2.5,-0.2+0.9) rectangle (0.4+2.5,0.2+0.9);
		\fill[blue!60!cyan!40!white] (-0.4+2.5,-0.2+2.7) rectangle (0.4+2.5,0.2+2.7);
		
		\fill[blue!60!cyan!40!white] (-0.4+3,-0.2) rectangle (0.4+3,0.2);
		\fill[blue!60!cyan!40!white] (-0.4+3,-0.2+1.8) rectangle (0.4+3,0.2+1.8);
		\fill[blue!60!cyan!40!white] (-0.4+3,-0.2+3.6) rectangle (0.4+3,0.2+3.6);
		
		\fill[blue!60!cyan!40!white] (-0.4+3.5,-0.2+0.9) rectangle (0.4+3.5,0.2+0.9);
		\fill[blue!60!cyan!40!white] (-0.4+3.5,-0.2+2.7) rectangle (0.4+3.5,0.2+2.7);
		
		\fill[blue!60!cyan!40!white] (-0.4+4,-0.2) rectangle (0.4+4,0.2);
		\fill[blue!60!cyan!40!white] (-0.4+4,-0.2+1.8) rectangle (0.4+4,0.2+1.8);
		\fill[blue!60!cyan!40!white] (-0.4+4,-0.2+3.6) rectangle (0.4+4,0.2+3.6);	
		
		\fill[blue!60!cyan!40!white] (-0.4+4.5,-0.2+0.9) rectangle (0.4+4.5,0.2+0.9);
		\fill[blue!60!cyan!40!white] (-0.4+4.5,-0.2+2.7) rectangle (0.4+4.5,0.2+2.7);
		
		\fill[blue!60!cyan!40!white] (-0.4+5,-0.2) rectangle (0.4+5,0.2);
		\fill[blue!60!cyan!40!white] (-0.4+5,-0.2+1.8) rectangle (0.4+5,0.2+1.8);
		\fill[blue!60!cyan!40!white] (-0.4+5,-0.2+3.6) rectangle (0.4+5,0.2+3.6);	
		
		\draw[line width=1 pt] (-0.4,-0.2) rectangle (0.4,0.2);
		\draw[line width=1 pt] (-0.4,-0.2+1.8) rectangle (0.4,0.2+1.8);
		\draw[line width=1 pt] (-0.4,-0.2+3.6) rectangle (0.4,0.2+3.6);
		
		\draw[line width=1 pt] (-0.4+0.5,-0.2+0.9) rectangle (0.4+0.5,0.2+0.9);
		\draw[line width=1 pt] (-0.4+0.5,-0.2+2.7) rectangle (0.4+0.5,0.2+2.7);
		
		\draw[line width=1 pt] (-0.4+1,-0.2) rectangle (0.4+1,0.2);
		\draw[line width=1 pt] (-0.4+1,-0.2+1.8) rectangle (0.4+1,0.2+1.8);
		\draw[line width=1 pt] (-0.4+1,-0.2+3.6) rectangle (0.4+1,0.2+3.6);
		
		\draw[line width=1 pt] (-0.4+1.5,-0.2+0.9) rectangle (0.4+1.5,0.2+0.9);
		\draw[line width=1 pt] (-0.4+1.5,-0.2+2.7) rectangle (0.4+1.5,0.2+2.7);
		
		\draw[line width=1 pt] (-0.4+2,-0.2) rectangle (0.4+2,0.2);
		\draw[line width=1 pt] (-0.4+2,-0.2+1.8) rectangle (0.4+2,0.2+1.8);
		\draw[line width=1 pt] (-0.4+2,-0.2+3.6) rectangle (0.4+2,0.2+3.6);
		
		\draw[line width=1 pt] (-0.4+2.5,-0.2+0.9) rectangle (0.4+2.5,0.2+0.9);
		\draw[line width=1 pt] (-0.4+2.5,-0.2+2.7) rectangle (0.4+2.5,0.2+2.7);
		
		\draw[line width=1 pt] (-0.4+3,-0.2) rectangle (0.4+3,0.2);
		\draw[line width=1 pt] (-0.4+3,-0.2+1.8) rectangle (0.4+3,0.2+1.8);
		\draw[line width=1 pt] (-0.4+3,-0.2+3.6) rectangle (0.4+3,0.2+3.6);
		
		\draw[line width=1 pt] (-0.4+3.5,-0.2+0.9) rectangle (0.4+3.5,0.2+0.9);
		\draw[line width=1 pt] (-0.4+3.5,-0.2+2.7) rectangle (0.4+3.5,0.2+2.7);
		
		\draw[line width=1 pt] (-0.4+4,-0.2) rectangle (0.4+4,0.2);
        \draw[line width=1 pt] (-0.4+4,-0.2+1.8) rectangle (0.4+4,0.2+1.8);
        \draw[line width=1 pt] (-0.4+4,-0.2+3.6) rectangle (0.4+4,0.2+3.6);	

        \draw[line width=1 pt] (-0.4+4.5,-0.2+0.9) rectangle (0.4+4.5,0.2+0.9);
        \draw[line width=1 pt] (-0.4+4.5,-0.2+2.7) rectangle (0.4+4.5,0.2+2.7);

        \draw[line width=1 pt] (-0.4+5,-0.2) rectangle (0.4+5,0.2);
        \draw[line width=1 pt] (-0.4+5,-0.2+1.8) rectangle (0.4+5,0.2+1.8);
        \draw[line width=1 pt] (-0.4+5,-0.2+3.6) rectangle (0.4+5,0.2+3.6);	

        \fill[blue!60!cyan!40!white] (-0.4,-0.2+0.9) rectangle (-0.1,0.2+0.9);
        \draw[line width=1 pt] (-0.4,-0.2+0.9) -- (-0.1, -0.2+0.9) -- (-0.1, 0.2+0.9) -- (-0.4, 0.2+0.9);

        \fill[blue!60!cyan!40!white] (-0.4,-0.2+2.7) rectangle (-0.1,0.2+2.7);
        \draw[line width=1 pt] (-0.4,-0.2+2.7) -- (-0.1, -0.2+2.7) -- (-0.1, 0.2+2.7) -- (-0.4, 0.2+2.7);

        \fill[blue!60!cyan!40!white] (5.1,-0.2+0.9) rectangle (5.4,0.2+0.9);
        \draw[line width=1 pt] (5.4,-0.2+0.9) -- (5.1, -0.2+0.9) -- (5.1, 0.2+0.9) -- (5.4, 0.2+0.9);

        \fill[blue!60!cyan!40!white] (5.1,-0.2+2.7) rectangle (5.4,0.2+2.7);
        \draw[line width=1 pt] (5.4,-0.2+2.7) -- (5.1, -0.2+2.7) -- (5.1, 0.2+2.7) -- (5.4, 0.2+2.7);

        \fill[white] (-0.25,0.45) circle (0.1);
        \draw[line width=1 pt] (-0.25,0.45) circle (0.1);

        \fill[white] (-0.25+1.5,0.45) circle (0.1);
        \draw[line width=1 pt] (-0.25+1.5,0.45) circle (0.1);

        \fill[white] (-0.25+4.5,0.45) circle (0.1);
        \draw[line width=1 pt] (-0.25+4.5,0.45) circle (0.1);

        \fill[white] (-0.25+1,0.45+0.9) circle (0.1);
        \draw[line width=1 pt] (-0.25+1,0.45+0.9) circle (0.1);

        \fill[white] (-0.25+4,0.45+0.9) circle (0.1);
        \draw[line width=1 pt] (-0.25+4,0.45+0.9) circle (0.1);

        \fill[white] (-0.25+1,0.45+1.8) circle (0.1);
        \draw[line width=1 pt] (-0.25+1,0.45+1.8) circle (0.1);

        \fill[white] (-0.25+3.5,0.45+1.8) circle (0.1);
        \draw[line width=1 pt] (-0.25+3.5,0.45+1.8) circle (0.1);

        \fill[white] (-0.25+1.5,0.45+2.7) circle (0.1);
        \draw[line width=1 pt] (-0.25+1.5,0.45+2.7) circle (0.1);

        \fill[white] (-0.25+2.5,0.45+2.7) circle (0.1);
        \draw[line width=1 pt] (-0.25+2.5,0.45+2.7) circle (0.1);

        \fill[white] (-0.25+5,0.45+2.7) circle (0.1);
        \draw[line width=1 pt] (-0.25+5,0.45+2.7) circle (0.1);

        \fill[white] (-0.25+0.5,0.45+3.6) circle (0.1);
        \draw[line width=1 pt] (-0.25+0.5,0.45+3.6) circle (0.1);

        \fill[white] (-0.25+3.5,0.45+3.6) circle (0.1);
        \draw[line width=1 pt] (-0.25+3.5,0.45+3.6) circle (0.1);

        \begin{scope}[xshift=7.8 cm, yshift=1.2 cm]
        \draw[line width=1 pt] (-0.25, -0.4) -- (-0.25, 0.4);
        \draw[line width=1 pt] (0.25, -0.4) -- (0.25, 0.4);
        \fill[blue!60!cyan!40!white] (-0.4,-0.2) rectangle (0.4,0.2);
        \draw[line width=1 pt] (-0.4,-0.2) rectangle (0.4,0.2);
        \node[] (C) at (0,-0.65) {random unitary};

        \begin{scope}[yshift= 2 cm]
        \draw[line width=1 pt] (0,-0.3) -- (0, 0.3);
        \fill[white] (0,0) circle (0.1);
        \draw[line width=1 pt] (0,0) circle (0.1);
        \node[] (C) at (0,-0.55) {$Z$-measurement};
        \end{scope}
        \end{scope}
        \end{scope}
		\end{tikzpicture}
		\caption{A spacetime diagram for the 1D randum Clifford (Haar) circuit model with projective measurements. The blue blocks are 2-qubit random unitary operators. Each white dot is a single-qubit projective measurement in the $Z$-basis. This type of measurements occurs with a probability $p$ in every discrete time step.}
		\label{fig:circuit}
	\end{figure}

	\subsubsection{Haar random unitary circuit with measurements}

	We also consider a more random evolution on the same physical system, namely the Haar random circuit with measurements. Now each blue block in Fig.~\ref{fig:circuit} refers to a 2-qubit unitary in $U(4)$ group, selected randomly with respect to the Haar measure of the Lie group $U(4)$. The measurements are still in the $Z$-basis as before, and these measurements remain independent with a probability $p$. The Haar random circuit model is much more computationally expensive, and we have simulated this model up to size $L=20$.

\subsection{Entanglement and correlation measures}

In this section, we briefly introduce the measures of entanglement and correlation studied in this work. Among them, entanglement entropy is an entanglement measure of bipartite pure states, mutual information is a measure of the total correlation for a bipartite mixed state \cite{2005PhRvA..72c2317G,2008PhRvL.100g0502W}, and the entanglement negativity is an entanglement measure of bipartite mixed states \cite{2006PhDT........59E,2009RvMP...81..865H,2017arXiv170102187D}.

\subsubsection{Entanglement entropy and mutual information}	

First we introduce the entanglement entropy and mutual information. The entanglement entropy, or von Neumann entropy, is defined as
  \begin{equation}\label{bipartite entropy}
  S(\rho_A)= -\Tr (\rho_A \log_2 \rho_A)
  \end{equation}
  where $\rho_A$ is the (reduced) density matrix on (sub)system $A$. We shall sometimes drop the state label and denote the entropy as $S_A$. Suppose the whole system is $A\cup B$. For a pure quantum state of the system, the von Neumann entropy $S_A$ characterizes the entanglement between $A$ and $B$. In other words, the von Neumann entropy is an entanglement measure for bipartite pure states. For a mixed state on $A\cup B$, on the other hand, the von Neumann entropy $S_A$ does not characterize the entanglement between $A$ and $B$ \cite{2006PhDT........59E,2009RvMP...81..865H,2017arXiv170102187D}.

  The \emph{mutual information} is defined as
  \begin{equation}\label{mutual information}
  I_{A,B}= S_A + S_B - S_{AB}.
  \end{equation}
  It is a measure of the total amount of correlations between the two subsystems $A$ and $B$, for possibly mixed quantum states~\cite{2005PhRvA..72c2317G,2008PhRvL.100g0502W}. Correlation can be either quantum or classical. For example, a \emph{separable state}
  \begin{equation}\label{separability}
  \rho_{AB}= \sum_i p_i \rho^i_A \otimes \rho^i_B
  \end{equation}
  can have classical correlation, but it has no quantum entanglement between $A$ and $B$ \cite{2006PhDT........59E,2009RvMP...81..865H,2017arXiv170102187D}. Here $\{\rho^i_A\}$ and $\{ \rho^i_B \}$ are density matrices and $\{p_i\}$ is a probability distribution.

    \subsubsection{Entanglement negativity}
    To quantify the entanglement between two subsystems for a possibly mixed quantum state, one may consider a bipartite mixed state entanglement measure. A number of such entanglement measures were introduced and studied; see \cite{2009RvMP...81..865H,2017arXiv170102187D} for a review. These measures must be entanglement monotones, i.e., they satisfy the condition to be non-increasing under local operations and classical communication (LOCC) between the two regions~\cite{2006PhDT........59E,2002PhRvA..65c2314V}. Furthermore, they vanish for a separable state (\ref{separability}). Therefore these quantities character quantum entanglement rather than the classical correlation for possibly mixed states. 
        Some of these bipartite entanglement measures have very nice theoretical properties. One such example is the squashed entanglement \cite{2004JMP....45..829C}, which is additive. However, calculating it requires a minimization process, and therefore it is hard to calculation in general.
    
    Among these proposed bipartite mixed state entanglement measures, the entanglement negativity is a calculable measure~\cite{2002PhRvA..65c2314V,2006PhDT........59E}. 
    The \emph{logarithmic entanglement negativity}\footnote{An alternative expression, which is called \emph{entanglement negativity} of state $\rho$, between regions $A$ and $B$, is defined as 
    	\begin{equation}
    	\mathcal{N}^{A\vert B}(\rho)= \frac{  \lVert \rho_{AB}^{T_A}  \rVert_1 -1 }{2}. \nonumber
    	\end{equation}
    Throughout the paper, we will stick to the logarithmic entanglement negativity instead of this one.} 
    for state $\rho$, between $A$ and $B$, is defined as
    \begin{equation}\label{log negativity}
    E_N^{A\vert B}(\rho)= \log_2  \lVert \rho_{AB}^{T_A}\rVert_1.
    \end{equation}
    where $T_A$ denotes the partial transpose with respect to region $A$ only, and $\lVert O \rVert_1$ is the trace norm of operator $O$. 
    For a more detailed description, please refer to Appendix~\ref{appendix:stabilizer}. 
    
    In the rest of the paper, we will study the logarithmic entanglement negativity in random unitary circuits with measurements. We shall call it entanglement negativity for short.
    
	\section{The critical behavior of entanglement negativity}\label{sec:results}

	In this section, we discuss our numerical results on the negativity $E_N^{A\vert B}$ as well as mutual information $I_{A,B}$ for the Clifford and Haar circuit models. All physical quantities are calculated as an ensemble average of that on possible late time quantum states. The goal is to understand quantum entanglement properties of the measurement-driven phase transition discovered in Refs.\cite{2018PhRvB..98t5136L,2019PhRvX...9c1009S,2019PhRvB..99v4307C,2019PhRvB.100m4306L}. Previously, numerical results \cite{2019PhRvB.100m4306L,2019PhRvB.100f4204S} provided excellent evidence that single interval entanglement entropy at the critical point scales with the logarithm of the subsystem size in both models.
	Here we divide the 1D system into two disjoint intervals $A$ and $B$, and the remaining region $\overline{A\cup B}$, and investigate the entanglement and correlation between the intervals $A$ and $B$ of the reduced density matrix $\rho_{AB}$ obtained by tracing out region $\overline{A\cup B}$. While the mutual information (\ref{mutual information}) has been studied previously by Ref.\cite{2019PhRvB.100m4306L}, here we focus on the logarithmic entanglement negativity (\ref{log negativity}), which is known to distinguish quantum entanglement from classical correlations between $A$ and $B$ for a mixed state.

To characterize the scaling behavior of the entanglement at the measurement-driven transition, we choose $A$ and $B$ to be disjoint arcs of a circle. We use $x_i \in [0,2\pi)$, $i=1,2,3,4$ to label the angular positions of the endpoints, as shown in Fig.~\ref{fig:cross_ratio}. The cross-ratio is defined as
		\begin{equation}
		\eta= \frac{x_{12} x_{34}}{x_{13} x_{24}}, \label{eq:cross_ratio}
		\end{equation}
		where $x_{ij}$ is the chord distance. For periodic boundary condition
		\begin{equation}
		x_{ij} = \frac{L}{\pi} \sin \left( \frac{\pi}{L} \vert x_i - x_j \vert \right).
		\end{equation}
We will investigate the scaling of entanglement negativity in the limit of a small cross-ratio $\eta\ll 1$. Previously, the scaling of mutual information on two disjoint intervals has been studied by Ref.~\cite{2019PhRvB.100m4306L}.

		\begin{figure}[h]
		\centering
		\begin{tikzpicture}
		\begin{scope}[xshift=2.6 cm, yshift=5.5 cm, scale=1.0]

		\draw[line width=1.5pt] (0,0) circle (1.6);
		\draw[line width=0.8 pt, dotted] (0,0)--(110:1.6);
		\draw[line width=0.8 pt, dotted] (0,0)--(176:1.6);
		\draw[line width=0.8 pt, dotted] (0,0)--(-85:1.6);
		\draw[line width=0.8 pt, dotted] (0,0)--(7:1.6);
		
		\draw[green!54!yellow!80!black,line width=3 pt] (110:1.6) arc (110:176:1.6);
		\draw[green!54!yellow!80!black,line width=3 pt] (-85:1.6) arc (-85:7:1.6);
		
		\node[] (C) at (143:1.92) {$A$};
		\node[] (C) at (-39:1.88) {$B$};
		
		\node[] (C) at (176:1.91) {$x_1$};
		\node[] (C) at (110:1.9) {$x_2$};
		\node[] (C) at (7:1.92) {$x_3$};
		\node[] (C) at (-85:1.905) {$x_4$};

		\end{scope}
		\end{tikzpicture}
		\caption{A possible choice of disjoint intervals $A$ and $B$. Also illustrated are the labels of the endpoints.  }
		\label{fig:cross_ratio}
	\end{figure}
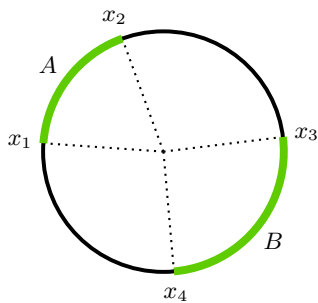

\subsection{Determining the critical measurement probability $p_c$}

		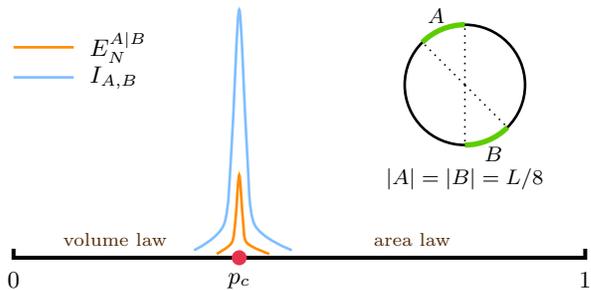
\begin{figure}[h]
	\centering
	\begin{tikzpicture}
	\begin{scope}[xshift=6.2 cm, yshift=0.3 cm, scale=0.50]

	\draw[line width=1pt] (0,0) circle (1.6);
	\draw[line width=0.6 pt, dotted] (0,0)--(90:1.6);
	\draw[line width=0.6 pt, dotted] (0,0)--(135:1.6);
	\draw[line width=0.6 pt, dotted] (0,0)--(-90:1.6);
	\draw[line width=0.6 pt, dotted] (0,0)--(-45:1.6);
	
	\draw[green!54!yellow!80!black,line width=2 pt] (90:1.6) arc (90:135:1.6);
	\draw[green!54!yellow!80!black,line width=2 pt] (-90:1.6) arc (-90:-45:1.6);
	
	\node[] (C) at (112.5:2) {\footnotesize{$A$}};
	\node[] (C) at (112.5+180:2) {\footnotesize{$B$}};
	
	\node[] (C) at (-90:2.5) {\footnotesize{$\vert A\vert = \vert B\vert = L/8$}};
	
	\end{scope}

	\begin{scope}[yshift=-2 cm]
	\draw[line width=1.5pt] (0.2,0.12)--(0.2,0) -- (7.8,0) -- (7.8, 0.12);
	
	\node[] (C) at (0.2,-0.3) {$0$};	
	\node[] (C) at (7.8,-0.3) {$1$};	
	
	\fill[red!50!purple!80!white] (3.2,0) circle (0.1);
	\node[] (C) at (3.2,-0.3) {$p_c$};	
	
	\node[] (C) at (1.55,0.25) {\color{red!50!yellow!30!black}\scriptsize{volume law}};
	\node[] (C) at (5.5,0.25) {\color{red!50!yellow!30!black}\scriptsize{area law}};
	
	\draw[red!45!yellow!99!white,line width=1 pt] (0.2,2.8)--(1,2.8);
	\node[right] (C) at (1.1, 2.8) {$E_N^{A\vert B}$ };
	
	\draw[ blue!50!cyan!50!white,line width=1 pt] (0.2,2.4)--(1,2.4);
	\node[right] (C) at (1.1, 2.4) {$I_{A,B}$ };
	
	\draw[red!45!yellow!99!white,line width=1 pt] plot [smooth] coordinates {(2.9, 0.05)(3.1, 0.2) (3.15, 0.5) (3.2, 1.1) (3.25, 0.5)  (3.3, 0.2) (3.6, 0.05) };
	\draw[blue!50!cyan!50!white,line width=1 pt] plot [smooth] coordinates {(2.9 -0.3, 0.05*2)(3.1 -0.1, 0.15*3) (3.15 -0.05, 0.5*3) (3.2, 1.1*3) (3.25+0.05 , 0.5*3)  (3.3+0.1, 0.15*3) (3.6+0.3, 0.05*2) };
	\end{scope}
	\end{tikzpicture}
	\caption{Schematic phase diagram of the random circuit model with projective measurements, as a function of measurement rate $p$. The volume law phase (left) and area law phase (right) are separated by the measurement-driven entanglement transition at $p_c$, the critical measurement rate. For relatively small disjoint intervals $A$ and $B$, fixed as two antipodal regions with $\vert A\vert = \vert B\vert = L/8$ ($\eta\approx 0.146$) in this diagram, both the mutual information and the entanglement negativity provide a sharp feature of the critical point. Both $I_{A,B}$ and $E_N^{A\vert B}$ are only nonzero nearby the critical point, and they quickly vanish upon entering either the volume law phase or the area law phase.}
	\label{fig:phase_diagram_cartoon}
\end{figure}

Previous studies~\cite{2018PhRvB..98t5136L,2019PhRvX...9c1009S,2019PhRvB..99v4307C,2019PhRvB.100m4306L} have established the phase diagram of the Clifford and Haar circuit models and identified the measurement-driven phase transition. There is a highly entangled phase, at a small measurement rate $p<p_c$ with volume law entanglement entropy for a single interval, and a disentangled phase at a large measurement rate $p>p_c$ with an area law entanglement entropy; see Fig.~\ref{fig:phase_diagram_cartoon} for an illustration. One can numerically determine the critical measurement probability $p_c$ for both models. 

Our strategy of determining $p_c$ is the based on the observation \cite{2019PhRvB.100m4306L} that the mutual information $I_{A,B}$ has a sharp peak at $p_c$ and it vanishes quickly with $|p-p_c|$ once entering the volume (area) law phase for a small cross-ratio $\eta$. This is illustrated by the blue curve in Fig.~\ref{fig:phase_diagram_cartoon}. We find that the entanglement negativity exhibits a similar peak at $p_c$, shown by the orange curve in Fig.~\ref{fig:phase_diagram_cartoon}. Using these sharp features, we can estimate $p_c$ for the measurement-driven phase transition.
The details of data collapse analysis to determine $p_c$ in the Clifford circuit model are presented in the Appendix~\ref{appendix:data_collaps_details}. 
The value of $p_c$, which we shall use in the remaining part of this paper, is:
\begin{itemize}
	\item $p_c\approx 0.16$ for the Clifford circuit model, consistent with Ref.~\cite{2019PhRvB.100m4306L}.
	\item $p_c\approx 0.26$ for the Haar circuit model with $L=20$, consistent with Ref.~\cite{2019PhRvX...9c1009S,2020PhRvB.101j4301B}.
\end{itemize}

 \subsection{Scaling of entanglement negativity at $p_c$}

 The main result of this work is a power-law dependence of entanglement negativity on the cross-ratio (\ref{eq:cross_ratio}), right at the critical point. We compare this power-law behavior with that of the mutual information; see Fig.~\ref{fig:data_main} for the plot. We vary the size and position of the interval when we collect data.  (Namely, we symmetrically adjust the position of the four points on the ring while keeping two intervals of equal length.)
In both the Clifford circuit model and the Haar circuit model, we obtain the same power-law behavior for two disjoint intervals $A$ and $B$ in Fig. \ref{fig:cross_ratio}:
\begin{itemize}
	\item $E_N^{A\vert B}\propto \eta^{\Lambda}$, where $\Lambda\approx 3$.
	\item $I_{A,B}\propto \eta^{\Delta}$, where $\Delta\approx 2$. This result agrees with Ref.~\cite{2019PhRvB.100m4306L}.	
\end{itemize}
More precisely, we find that 
\begin{itemize}
	\item $\Lambda_{\textrm{Clifford}}=3.04\pm 0.08$ and $\Lambda_{\textrm{Haar}}=2.73 \pm 0.28$.
	\item $\Delta_{\textrm{Clifford}}=2.16\pm 0.03$ and $\Delta_{\textrm{Haar}}=1.89 \pm 0.26$.
\end{itemize}
Here the error bars come from the standard deviation of the linear fit.

	\begin{figure}[h]
		\includegraphics[width=0.8\columnwidth]{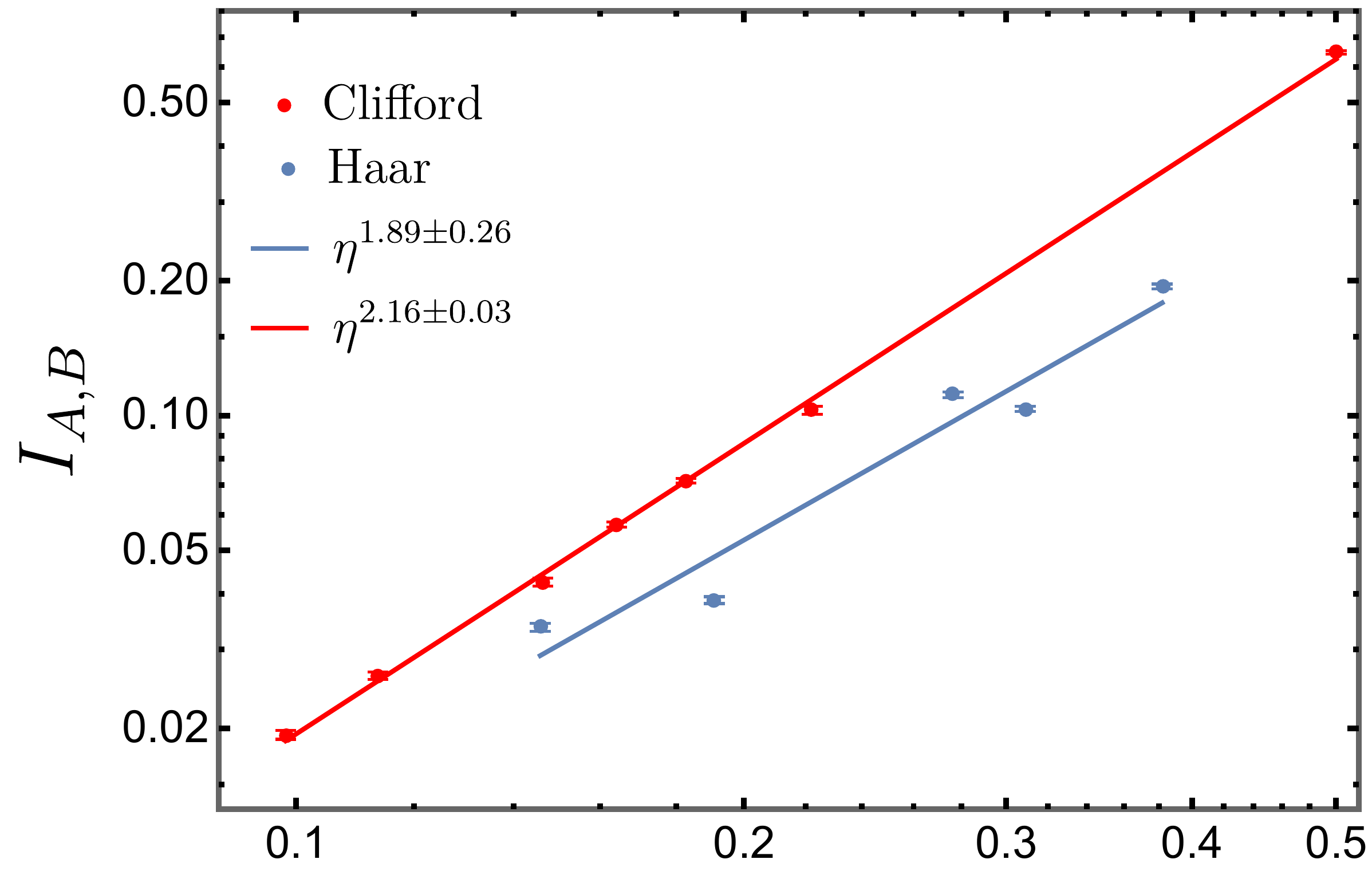}
		\includegraphics[width=0.82\columnwidth]{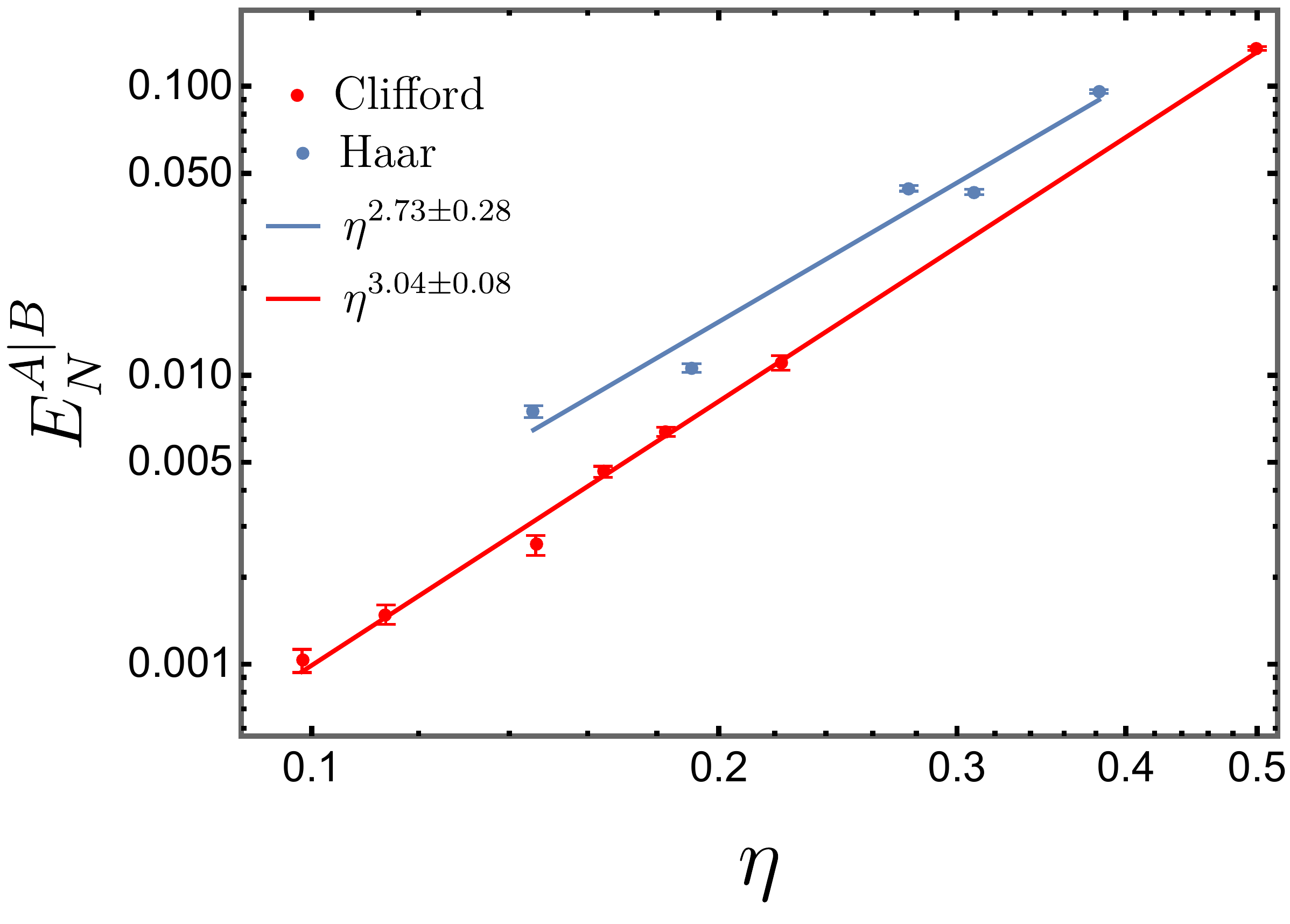}
		\caption{Mutual information $I_{A,B}$ and (logarithmic) entanglement negativity $E_N^{A\vert B}$ scales as a power of cross-ratio $\eta$ at the critical point, for small $\eta$. (We take $p_c=0.16$ and $L=512$ for Clifford and $p_c=0.26$ and $L=20$ for Haar.) The first figure suggests that the power-law scaling exponent of the mutual information are $\Delta_{\textrm{Clifford}}=2.16\pm 0.03$ and $\Delta_{\textrm{Haar}}=1.89 \pm 0.26$. The second figure suggests that the scaling exponent of the (logarithmic) entanglement negativity are $\Lambda_{\textrm{Clifford}}=3.04\pm 0.08$ and $\Lambda_{\textrm{Haar}}=2.73 \pm 0.28$.} \label{fig:data_main}
	\end{figure}

Both the entanglement negativity and the mutual information scale as a power of the cross-ratio, albeit that the power $\Lambda\approx 3$ is different from $\Delta\approx 2$.
The mutual information characterizes the total amount of correlation between two regions, including both classical correlation and quantum entanglement. On the other hand, the entanglement negativity quantifies the amount of quantum entanglement between two regions.
The scaling dimension $\Lambda > \Delta$ physically means that the quantum entanglement between two disjoint intervals drops faster than the correlation when the distance between the intervals increases. It is physically reasonable to expect $\Lambda \geq\Delta$, because the amount of quantum entanglement must be bounded by the total amount of correlations.\footnote{Depending on the specific entanglement measure in question, the bound may not look simple. However, in the context of the stabilizer states, which is relevant to the random Clifford circuit model with measurement, a simple bound exists: $I_{A,B} \ge 2 E_N^{A\vert B}$; see Theorem~\ref{Thm:NE_VS_Mutaul_stabilizer} in Appendix~\ref{appendix:stabilizer}. }

Can these critical phenomena be captured by an effective field theory? As observed in previous studies\cite{2019PhRvB.100m4306L,2020arXiv200312721L,2020PhRvB.101j4302J}, the power-law behavior ($I_{A,B}\propto \eta^{\Delta}$) of the mutual information (also see Ref.~\cite{2019PhRvX...9c1009S} for the 0-th R\'{e}nyi mutual information) suggests an emergent conformal field theory (CFT) description of the critical point. One way this might work is to match the entanglement behavior of the measurement-driven transition to that of some CFT ground-state. For a unitary CFT ground state, the mutual information of two disjoint intervals is known to depend on the full operator content of the theory \cite{2009JSMTE..11..001C,2011JSMTE..01..021C,2011JSMTE..01..021C}, and calculation must be done case by case.
Nevertheless, if the CFT is compact and unitary, the 2nd Renyi mutual information is rigorously shown to be a power of $\eta$ at small $\eta$, and the power is determined knowing the dimension of the lowest non-unit operator in the CFT~\cite{2010PhRvD..82l6010H}. Therefore, the scaling of mutual information is analogous to that for a unitary CFT ground state. The power constrains the nature of any CFT that may describe the measurement-driven transition.

The same question arises for entanglement negativity: what does the scaling of negativity imply about the nature of the measurement-drive transition? As a main result of this paper, we numerically observe a power-law dependence of the logarithmic entanglement negativity (\ref{log negativity}) on the cross-ratio (\ref{eq:cross_ratio}) at the critical point ($E_N^{A\vert B}\propto \eta^{\Lambda}$). For unitary CFT ground states, the entanglement negativity is known to decay faster than any power law of cross-ratio for small $\eta$~\cite{2012PhRvL.109m0502C,2013JSMTE..02..008C,2015JSMTE..06..021D}. Thus our result suggests that the entanglement behavior at the critical point of the measurement-driven transition is distinct from that of any unitary CFT ground state.  

Nonunitary CFTs \cite{2005ffsc.book.1384G,2006hep.th....7232P} are natural candidates of the low-energy effective field theory because they obey the conformal symmetry, and the behavior of entanglement negativity may not obey the same rules as the unitary CFTs~\cite{2012PhRvL.109m0502C,2013JSMTE..02..008C}, as discussed in Ref.~\cite{2018ScPP....4...31D} and the references therein. However, a careful future study is needed to determine the relevance and the nature of possible emergent nonunitary CFT effective description of the critical point in the Clifford and Haar circuit models. The scaling exponent $\Lambda\approx 3$, can provide useful information in constraining the nature of this non-unitary CFT.

\subsection{Different ``Quantumness" of Haar vs. Clifford circuit models}

The Clifford circuit model and the Haar circuit model have almost identical scaling exponents within the error bar, for both the mutual information and the entanglement negativity, as shown in Fig.~\ref{fig:data_main}. However, we have observed a slight difference in the ``quantumness'' of these two models, which we describe below.    

For a generic mixed quantum states, while mutual information (\ref{mutual information}) measures the total amount of (both classical and quantum) correlation between two disjoint intervals $A$ and $B$, the entanglement negativity vanishes for any separable states and therefore characterizes the quantum entanglement between the two intervals. Consequently the ratio of the negativity to the mutual information provides a measure of how much correlation comes from quantum entanglement between two regions $A$ and $B$, i.e. the ``quantumness'' of the state. To compare the Clifford and the Haar circuit models, we define the following quantity:
\begin{equation}
R(\eta) \equiv \left(\frac{E_N^{A\vert B}(\eta)}{I_{A,B}(\eta)} \right)_{\textrm{Haar}} \left/ \left(\frac{E_N^{A\vert B}(\eta)}{I_{A,B}(\eta)} \right)_{\textrm{Clifford}} \right. \label{quantum ratio}
\end{equation}
to feature the excess amount of ``quantumness" in the Haar circuit model as compared to Clifford circuit model. $\eta$ is the cross-ratio defined in (\ref{eq:cross_ratio}).

Our numerical data suggests that $R(\eta)$ depends weakly on the cross-ratio $\eta$:
\begin{equation}
R(\eta)=R_0\eta^{\delta}, \quad R_0=2.94, \, \delta = -0.04 \pm 0.39.
\end{equation}
$R(\eta)\approx 3$ in the whole range of cross-ratio that we simulate. While this does not imply any difference in the effective field theory description of the two models, this does suggest that there is a larger fraction of quantum entanglement in the Haar circuit model than the Clifford circuit model. Therefore the Haar circuit model appears to be more quantum than the Clifford circuit model at the measurement-driven transition.

	\section{Summary and discussion}\label{sec:summary}
	
In this work, we have numerically studied the correlation and entanglement between two disjoint intervals at the measurement-driven phase transition in 1D Clifford and Haar random unitary circuits. We focus on the mutual information (\ref{mutual information}) and the logarithmic entanglement negativity (\ref{log negativity}). The late time quantum states are generically mixed on the union of the two intervals. Therefore, the mutual information measures the total correlation between the two intervals, whereas the entanglement negativity characterizes the quantum entanglement between the two intervals. We numerically simulate 1D systems with a periodic boundary condition, up to a system size $L=768$ for the Clifford circuit model and $L=20$ for the Haar circuit model. We identified a power-law behavior of the (logarithmic) entanglement negativity at small interval sizes. For a pair of disjoint intervals $A$ and $B$,
	\begin{equation}
	E_N^{A\vert B} \propto \eta^{\Lambda} \quad \textrm{with }\, \Lambda\approx 3 \label{eq:power_NE_summary}
	\end{equation}
	 at $\eta\ll 1$, where $\eta$ is the cross-ratio for the pair of disjoint interval defined in (\ref{eq:cross_ratio}).
	
	It is interesting to compare (\ref{eq:power_NE_summary}) with recently obtained power-law behavior of the mutual information~\cite{2019PhRvB.100m4306L}.
	\begin{equation}
	I_{A,B} \propto \eta^{\Delta} \quad \textrm{with }\, \Delta\approx 2 \label{eq:power_mutual_summary}
	\end{equation}
	at $\eta\ll 1$. The fact that the scaling dimension $\Lambda > \Delta$ suggests that the quantum entanglement between two disjoint intervals decays faster than the correlation as we increase the distance between the intervals. Moreover, we have observed that the Haar circuit model is ``more quantum'' than the Clifford circuit model, by comparing the ratio (\ref{quantum ratio}) of the entanglement negativity to mutual information in the two models.

As the main conclusion of this work, the power-law dependence of negativity on the cross-ratio indicates that the description of the measurement-driven transition in random unitary circuits is distinct from the ground state of any unitary conformal field theory. 
Previous studies on the scaling of single interval von Neumann and Renyi (entanglement) entropies, as well as the mutual information for two disjoint intervals, point to an emergent conformal invariance of the critical point\cite{2020arXiv200312721L}. However, the precise nature of the CFT describing the phase transition is unclear\cite{2019PhRvB.100m4306L,2020arXiv200312721L}. Remarkably, unlike the power-law mutual information between two disjoint intervals, in a unitary CFT ground state, the negativity always decays faster than any power law of cross-ratio $\eta$ for small $\eta$~\cite{2012PhRvL.109m0502C,2013JSMTE..02..008C}. Therefore, the power-law behavior of entanglement negativity provides direct evidence that the entanglement property of the measurement-driven transition is distinct from that of the ground state of any unitary CFT.

Conformal symmetry may, instead, enter the story through statistical mechanical models. In such cases, the resulting property does not necessarily match the ground state properties~\cite{2008PhRvL.100h7205J}.
Previously, Vasseur et al. \cite{2019PhRvB.100m4203V} argued a vanishing central charge ($c=0$) for the CFT corresponding to the measurement-driven transition, and indicated that this CFT belongs to a class of non-unitary Logarithmic CFT. This is achieved by using the replica trick and mapping the entanglement entropy in a random unitary circuit with measurements to the change of free energy of a 2D statistical mechanical model w.r.t. twisting the boundary condition\cite{2018PhRvX...8b1014N,2020PhRvB.101j4301B,2020PhRvB.101j4302J,2020arXiv200212385F}. The critical point of the statistical mechanical model belongs to the $(Q!)$-state Potts model in the $Q\rightarrow1$ limit, in the universality class of 2D percolation, perturbed by a relevant 2-hull operator. In the current work, inspired by recent numerical evidence of emergent conformal invariance in the mutual information, we assume the critical behavior of the measurement-driven transition is described by a CFT ground state, analogous to continuous phase transitions in equilibrium. The power-law behavior of entanglement negativity hence suggests this CFT cannot be unitary. Currently, it is not clear how to exactly map the entanglement negativity to a 2D statistical mechanics model, employing the idea of Ref.~\cite{2019PhRvB.100m4203V}. This is an interesting question to be addressed in the future.

While the power-law negativity in this work provides an extra constraint on the critical theory, many questions remain open for this entanglement phase transition. For example, how to extract other critical exponents of this dynamical critical point? What is the precise nature of the conformal symmetry that emerged at the critical point of the measurement-driven phase transition? How to understand the scaling power that we observed? We leave these questions to future works.

	\section*{Acknowledgments}
	We thank Tarun Grover, Matthew Fisher, Brian Skinner, Sagar Vijay, Xueda Wen, Yahya Alavirad, John McGreevy, Yi-Zhuang You for interesting discussions. We are especially grateful to Romain Vasseur for his helpful comment on the nature of the conformal field theory in statistical mechanics models. This work is supported by the National Science Foundation under Grant No. NSF DMR-1653769 (BS, XD, YML), University of California Laboratory Fees Research Program, grant LFR-20-653926 (BS), and the Simons Collaboration on Ultra-Quantum Matter, grant 651440 from the Simons Foundation (BS).
	
	\emph{Note: During the preparation of this paper, we became aware of an independent work by Shengqi Sang, Yaodong Li, Tianci Zhou, Xiao Chen, Timothy H. Hsieh, and Matthew P. A. Fisher, who also studied entanglement negativity at the same measurement-driven transition in random unitary circuits (Ref.~\cite{Sang2020}).}
	
	\appendix

	\section{Entanglement of stabilizer states}\label{appendix:stabilizer}

In this appendix, we provide some details of the entanglement of stabilizer states. In Section~\ref{ap:stabilizer_state}, we provide a brief review of the stabilizer states and set up the notation for the remaining discussion. In Section~\ref{ap:stabilizer_entropy}, we review the calculation of entanglement entropy (bipartite pure state entanglement measure) of the stabilizer states. In Section~\ref{ap:NE_1}, \ref{ap:NE_2} and \ref{Sec:Proof}, we describe a method to calculation entanglement negativity (a bipartite mixed state entanglement measure) of stabilizer states. While this appendix is relatively self-contained, we mention that various statements in this appendix may be seen from a useful alternative viewpoint on the multipartite entanglement of stabilizer states \cite{2006JMP....47f2106B}.

As a reminder, Clifford unitary operators and projective measurements in the $Z$-basis, when applied to a stabilizer state, corresponds to an update of the stabilizer generators. The state after each step of evolution is again a stabilizer state. Therefore, the states we consider in the random Clifford circuit model with measurements are stabilizer states. For a review of how to update the stabilizers for Clifford unitary evolution and projective measurement, see \cite{1998quant.ph..7006G,2004PhRvA..70e2328A} and appendices of  \cite{2019PhRvB.100m4306L} and \cite{2017PhRvX...7c1016N}.

\subsection{Stabilizer states}\label{ap:stabilizer_state}
We will consider a quantum system built up from a set of qubits. The number of qubit of the whole system is $L$. In other words, the total Hilbert space is 
\begin{equation}
\mathcal{H}= \mathcal{H}_2^{\otimes L},\quad \mathcal{H}_2=span\{ \vert 0\rangle, \vert 1\rangle \}. 
\end{equation}

Below, we define three concepts: stabilizer group, stabilizer generators, and stabilizer states.

\begin{definition} (stabilizer group)
	The stabilizer group $\mathcal{G}$ is defined to be the Abelian group generated by the following operators $\{ g_1, g_2,\cdots , g_L \}$ acting on $\calH$. These operators, which we shall call stabilizer generators, satisfy:
	\begin{enumerate}
		\item Each $g_i$ is a product of Pauli operators. In other words, it is a product of operators acting on individual sites, and on each site, the operator (up to a phase factor) can be chosen from the set $\{ I,X,Y,Z \}$. Here $I$ is the identity, $X$, $Y$, $Z$ are single-qubit Pauli operators.    
		\item $g_i^{\dagger}= g_i$ and $g_ig_j =g_j g_i$, $\forall i,j$. 
		\item The set of operators are independent, i.e. 
		\begin{equation}
		g_i \ne \prod_{j\ne i} g_j^{s_j},\quad s_j =0,1
		\end{equation}
	\end{enumerate}
\end{definition}
\begin{remark}
	The definition implies $g_i^2=1$. It follows from the definition that $\vert \mathcal{G}\vert = 2^L$, where $\vert \mathcal{G}\vert$ is the number of group element.
	Due to the independence condition, the stabilizer generators can independently take $\pm 1$. For each choice, there is a unique quantum state (up to the overall phase factor) that has these as the eigenvalues. 
\end{remark}

\begin{definition}(stabilizer Hamiltonian)
	We define the following Hamiltonian as the stabilizer Hamiltonian of $\mathcal{G}$:
	\begin{equation}
	H=-\sum_{j=1}^L \,g_j, \label{eq:stabilizer_Hamiltonian}
	\end{equation}
	where $\{ g_1, g_2,\cdots , g_L \}$ is the set of stabilizer generators.	
\end{definition}

\begin{definition} (stabilizer state)
	The stabilizer state of $\mathcal{G}$ is the unique quantum state (up to an overall phase factor) $\vert \psi \rangle$ that satisfies
	\begin{equation}
	h \vert \psi\rangle = \vert \psi\rangle ,\quad  \forall h\in\mathcal{G}.
	\end{equation}
\end{definition}
\begin{remark}
	The stabilizer state is the unique ground state of the stabilizer Hamiltonian (\ref{eq:stabilizer_Hamiltonian}). 
\end{remark}

We define the following projector:
\begin{equation}
P \equiv  \frac{1}{\vert \mathcal{G} \vert } \sum_{ h\in \mathcal{G}} h\quad \Leftrightarrow \quad P =  \frac{1}{\vert \mathcal{G} \vert } \prod_{j=1}^L (1+g_j).
\end{equation}
It satisfies
\begin{eqnarray}
P&=&P^{\dagger}\\
P^2 &=& P.
\end{eqnarray}
In fact, $P$ is the projector to the ground subspace of the stabilizer Hamiltonian and therefore
\begin{equation}
P =\vert \psi\rangle \langle \psi\vert.
\end{equation}
In this way, we get a neat formula for the ground state density matrix.

Let us consider a subsystem $A$ of the system and let the system be $AB$. Let $L_A$ be the number of qubits in $A$.
Let $\mathcal{G}_{A}$  be the subgroup of stabilizers supported on $A$. See the formal definition below.
\begin{definition}
	\begin{equation}
	\mathcal{G}_{A}\equiv \{ h_A \, \vert h_A \otimes 1_B \in \mathcal{G} \}.
	\end{equation} 
\end{definition}
It is easy to say $\mathcal{G}_A$ is isomorphic to a subgroup of $\mathcal{G}$.

The reduced density matrix of the stabilizer state on subsystem $A$, $\rho_A \equiv \Tr_B \vert \psi\rangle \langle \psi\vert$, can be written as
\begin{equation}
\rho_A =\frac{ \vert \mathcal{G}_A\vert  }{ 2^{L_A}} P_A,\quad P_A\equiv \frac{1}{\vert \mathcal{G}_A\vert } \sum_{h_A \in  \mathcal{G}_A } h_A \label{eq:rho_A}
\end{equation}
or simply
\begin{equation}
\rho_A =\frac{ 1 }{ 2^{L_A}} \sum_{h_A \in \mathcal{G}_A} h_A.
\end{equation}

One easily verifies that $P_A$ in (\ref{eq:rho_A}) is a projector, and therefore
\begin{equation}
\rho_A^2 = \frac{ \vert \mathcal{G}_A\vert  }{ 2^{L_A}} \rho_A . \label{eq: flat}
\end{equation}
A density matrix satisfying (\ref{eq: flat}) must have a flat entanglement spectrum. In other words, all the nonzero eigenvalues of $\rho_A$ are the same: $\lambda= { \vert \mathcal{G}_A\vert  }/{ 2^{L_A}}$. Furthermore, $\vert \mathcal{G}_A\vert = 2^{m_A}$ for some nonnegative integer $m_A$. 

\subsection{Entanglement entropy of stabilizer states}\label{ap:stabilizer_entropy}
With the knowledge of the stabilizer states discussed above, it is straightforward to calculation its entanglement entropy. For a stabilizer state $\rho_A$, defined in (\ref{eq:rho_A}), we have
\begin{equation}
S(\rho_A) = S_{\alpha}(\rho_A) = (L_A - m_A), \quad \alpha\in [0,1)\cup (1,\infty) \label{eq:Renyi}
\end{equation}
Here $S(\rho_A)=-\Tr(\rho_A \log_2 \rho_A)$ is the von Neumann entropy and $S_{\alpha}$ is the Renyi entropy of order $\alpha$, defined as
\begin{equation}
S_{\alpha}(\rho_A) =\frac{1}{1-\alpha} \log_2 [\Tr\, (\rho_A)^{\alpha}].
\end{equation}
Note that Renyi entropy is not defined for $\alpha =1$. Nevertheless, Renyi entropy at the limit $\alpha \to 1$ (in both $1^{+}$ and $1^{-}$ direction) gives us the von Neumann entropy.
As an aside, we can infer from (\ref{eq:Renyi}) that $m_A\le L_A$.

A practically useful method is to write the binary vectors of a known set of stabilizer generators $\{ g_1, \cdots, g_L\}$ as a $L\times 2L$ matrix over $F_2$:
\begin{equation}
M= (M_A \vert M_B),
\end{equation}
where $M_A$ is a $L\times 2L_A$ matrix and $M_B$ is a $L\times 2 L_B$ matrix; we have partitioned the systems into $A$ and $B$ of length $L_A$ and $L_B= L-L_A$. Here each stabilizer generator is mapped to a row of length $2L$. (Explicitly,  we shall choose the binary representations $I=(00)$, $X=(10)$, $Z=(01)$, $Y=(11)$.)

We have the formula:
\begin{equation}
m_A= L_{AB}- \Rk(M_B),
\end{equation}
where the rank $\Rk(M_B)$ is defined over field $F_2$.

It is easy to see, the von Neumann entropy of a stabilizer state, $\rho_A$ of (\ref{eq:rho_A}), is
\begin{equation}
S(\rho_A) = \Rk(M_A) - L_A. \label{eq:S_A_practical}
\end{equation}

\begin{exmp}\label{Example:EE}
	Let $L=2$ and the set of stabilizer generators be $\{XX, ZZ\}$. Then the matrix $M$ reads
	\begin{equation}
	M=\left( \begin{array}{cccc}
	1&0&1&0\\0&1&0&1
	\end{array} \right). 
	\end{equation}
	Let $L_A=1$, then
	\begin{equation}
	M_A =\left( \begin{array}{cc}
	1&0\\0&1
	\end{array} \right) \quad \Rightarrow\quad \Rk(M_A)=2.
	\end{equation} 
	According to (\ref{eq:S_A_practical}), we have $S(\rho_A)=1$.
\end{exmp}

\subsection{Entanglement negativity basics}\label{ap:NE_1}

Entanglement negativity is a bipartite mixed state entanglement measure. While it goes back to (Renyi $\alpha=1/2$) entanglement entropy for a pure state, for a mixed state it is in general different from any linear combinition of entanglement entropies. It characterizes the quantum entanglement between two regions $A$ and $B$ for a possibly mixed state $\rho_{AB}$, rather than the correlations between $A$ and $B$. For example, we may consider a \emph{separable} state
\begin{equation}
\rho_{AB}= \sum_i p_i \, \rho^i_A \otimes \rho_B^i
\end{equation}
where $\{p_i\}$ is a probability distribution and $\rho^i_A $, $\rho^i_B$ are density matrices.
In general, a separable state can have nontrivial correlation between the two subsystems because, it is possible that the \emph{mutual information}
\begin{equation}
I_{A,B} (\rho_{AB})\equiv S(\rho_A) + S(\rho_B) - S(\rho_{AB}) > 0.
\end{equation}
However, for a separable state, the correlation is classical. This correlation can be generated from local operation and classical communication (LOCC). Therefore, it cannot be used as a resource for quantum teleportation. 

As a quantity non-increasing under LOCC \cite{2002PhRvA..65c2314V},  (logarithmic) entanglement negativity can characterize the quantum entanglement between $A$ and $B$, for a mixed state $\rho_{AB}$.

\begin{definition} (logarithmic entanglement negativity)
	For $\rho_{AB}$, let us define the logarithmic entanglement negativity as
	\begin{equation}
	E_N^{A\vert B}(\rho_{AB})= \log_2  \lVert \rho_{AB}^{T_A}\rVert_1.
	\end{equation}
\end{definition}

Below, we explain what is $T_A$ and the trace norm $\lVert A \rVert_1$.

To define $T_A$, we need to specify an orthonormal basis of subsystem $A$, and the operation $T_A$ will (in general) be different if we pick a different	basis. (Nevertheless, one can show that different $T_A$ will result in the same trace norm $\lVert \rho_{AB}^{T_A}\rVert_1$. For this reason, we do not need to care about the basis choice.

Let us pick an orthonormal  basis $\{ \vert i_A\rangle \}$ of $\mathcal{H}_A$ for $T_A$ and for convenience, we also pick an orthonormal basis $ \{ \vert j_B \rangle  \}$ for $\mathcal{H}_B$. Then, $T_A$ can be defined as the linear transformation on the operators acting on $\mathcal{H}_{AB}$ that satisfies 
\begin{equation}
(\vert i_A, j_B\rangle \langle i'_A, j'_B\vert )^{T_A}= \vert i'_A, j_B\rangle \langle i_A, j'_B\vert .
\end{equation}

The trace norm $\lVert A \rVert_1$ is defined as 
\begin{equation}
\lVert A \rVert_1 \equiv \Tr\sqrt{A^{\dagger} A}.
\end{equation}
When $A$ is a Hermitian operator, $\lVert A \rVert_1$ equals to sum of the absolute values of the eigenvalues of $A$. 
Note that $\rho_{AB}^{T_A}$ is a Hermitian operator when $\rho_{AB}$ is a density matrix.

\begin{remark}
	If $\rho_{AB}= \vert \psi_{AB}\rangle \langle \psi_{AB}\vert$ is a pure state density matrix, then 
	\begin{equation}
	E_{N}^{A\vert B}(\rho_{AB})= S_{\frac{1}{2}}(\rho_{AB}).
	\end{equation}
	Entanglement negativity is more interesting if we consider a tripartite pure state $\vert \psi_{ABC}\rangle$ or consider a mixed state $\rho_{AB}$.
\end{remark}

\subsection{Entanglement negativity of stabilizer states}\label{ap:NE_2}
Let $\vert \psi_{ABC}\rangle$ be a stabilizer state on a tripartite system $ABC$.
We trace out the subsystem $C$ and get 
\begin{equation}
\rho_{AB} =\frac{ \vert \mathcal{G}_{AB} \vert  }{ 2^{L_{AB}}} P_{AB},\quad P_{AB}\equiv \frac{1}{\vert \mathcal{G}_{AB}\vert } \sum_{h_{AB} \in  \mathcal{G}_{AB} } h_{AB}. \label{eq:rho_AB}
\end{equation}
This fact has been discussed around (\ref{eq:rho_A}).

Let us denote the set of stabilizer generators of the subgroup $\mathcal{G}_{AB}$ as $\{ h_{AB}^i \}_{i=1}^{m_{AB}}$, where the integer $m_{AB}$ is the number of stabilizer generators of $\mathcal{G}_{AB}$.
Then, a question is how to calculate $E_{N}^{A\vert B}({\rho_{AB}})$ efficiently given these stabilizer generators.

We find that it is useful to define a $m_{AB}\times m_{AB}$ symmetric matrix $J$ over field $F_2$:
\begin{equation}
J_{ij}= \left\{ \begin{array}{ll}
1 & \{ h^i_{A}, h^j_{A}\}=0,\\
0 & \textrm{otherwise,}
\end{array} \right. \label{eq:def_J}
\end{equation}
where we have factorized the stabilizer generators as
\begin{equation}
h_{AB}^i = h_A^i \otimes h_B^i.
\end{equation}
Note that, $J_{ii}=0$ for all $i$. Therefore, $J$ may also be treated as a skew-symmetric matrix over $F_2$.

\begin{theorem}\label{Thm:NE_stabilizer}
	For a stabilizer density matrix $\rho_{AB}$ in (\ref{eq:rho_AB}), 
	\begin{equation}
	E_N^{A\vert B} (\rho_{AB})= \frac{1}{2} \Rk(J),
	\end{equation}
	where $\Rk(J)$ is the rank of $J$ over field $F_2$.
\end{theorem}

Theorem~\ref{Thm:NE_stabilizer} is useful in the calculation of the entanglement negativity of stabilizer states. See Sec.~\ref{Sec:Proof} below for the proof.

\begin{theorem}\label{Thm:NE_VS_Mutaul_stabilizer}
	For a stabilizer density matrix $\rho_{AB}$ in (\ref{eq:rho_AB}),
	\begin{equation}
	2 E_N^{A\vert B} (\rho_{AB})\le I_{A,B} (\rho_{AB}). \label{eq:bound}
	\end{equation}
\end{theorem}
\begin{remark}
The same bound does not generalize to arbitrary quantum states. A simple way to see this is to observe that we can find a pure state $\vert \varphi_{AB}\rangle$ on a 2-qubit system such that $S_{\frac{1}{2}}(\rho_A)\ge \lambda S(\rho_A)$ for any real number $\lambda$. Here $\rho_A$ is the reduced density matrix of $\vert \varphi_{AB}\rangle$.
\end{remark}

\begin{proof}
	We only need to show
	\begin{equation}
	\Rk (J) \le m_{AB} - m_A - m_B.
	\end{equation}
	This is because we can rewrite the left-hand side of (\ref{eq:bound}) using Theorem~\ref{Thm:NE_stabilizer} and rewrite the right-hand side using (\ref{eq:Renyi}).
	
	For the calculation of $\Rk(J)$, we have the freedom to choose the set of stabilizer generators of $\mathcal{G}_{AB}$. Let us choose the following set
	\begin{equation}
	\{ h_A^i \otimes 1_B \}_{i=1}^{m_A} \cup \{ 1_A \otimes h_B^j \}_{j=1}^{m_B} \cup \{ l_A^k \otimes r_B^k\}_{k=1}^{m_{AB}-m_A -m_B}, \label{eq:good_basis}
	\end{equation}
	where $\{h_A^i\}_{i=1}^{m_A}$ is the set of generators of $\mathcal{G}_A$ and $\{h_B^j\}_{j=1}^{m_B}$ is the set of generators of $\mathcal{G}_B$. We have factorized the stabilizer generators of $\mathcal{G}_{AB}$ into a product on $A$ and $B$.  
	
	In the basis (\ref{eq:good_basis}), the matrix $J$ takes a simple form
	\begin{equation}
	J= \left( \begin{array}{c}
	0
	\end{array} \right)^{\oplus m_a} \oplus \left( \begin{array}{c}
	0
	\end{array} \right)^{\oplus m_b} \oplus J', \label{eq:key}
	\end{equation}
	where $J'$ is a symmetric $(m_{AB}-m_A -m_B)\times (m_{AB}-m_A -m_B)$ matrix over field $F_2$, defined according to
	\begin{equation}
	J'_{kk'}= \left\{ \begin{array}{ll}
	1 & \{ l^k_{A}, l^{k'}_{A}\}=0,\\
	0 & \textrm{otherwise.}
	\end{array} \right.
	\end{equation}
	
	It is obvious from (\ref{eq:key}) that $\Rk(J)\le  m_{AB} - m_A - m_B$. This completes the proof.
\end{proof}

Below are two simple examples that illustrate the calculation.

\begin{exmp}\label{Example:NE}
	Let $L_A=L_B=1$ and the set of stabilizer generators of $\mathcal{G}_{AB}$ be $\{ XX, ZZ\}$. Then the matrix $J$ is $2\times 2$:
	\begin{equation}
	J=\left( \begin{array}{cc}
	0&1\\
	1&0
	\end{array} \right) \quad \Rightarrow\quad \Rk(J)=2.
	\end{equation}
	Therefore, according to Theorem~\ref{Thm:NE_stabilizer}, 	$E_N^{A\vert B} (\rho_{AB})=1$. This result makes sense because the stabilizer state for this case is a Bell state
	\begin{equation}
	\vert \psi_{AB}\rangle = \frac{1}{\sqrt{2}} (\vert 00\rangle + \vert 11\rangle),
	\end{equation}
	written in the $Z$-basis.
\end{exmp}

\begin{exmp}\label{Example:NE_2}
	Let $L_A=L_B=L_C=1$ and the set of stabilizer generators of $\mathcal{G}$ be $\{ ZZI, IZZ, XXX \}$. Then the subgroup $\mathcal{G}_{AB}$ has a unique generator $ZZ$, and therefore the matrix $J$ is $1\times 1$:
	\begin{equation}
	J=\left( \begin{array}{c}
	0
	\end{array} \right) \quad \Rightarrow\quad \Rk(J)=0.
	\end{equation}
	Therefore, according to Theorem~\ref{Thm:NE_stabilizer}, 	$E_N^{A\vert B} (\rho_{AB})=0$. This result makes sense because the stabilizer state for this case is the GHZ state
	\begin{equation}
	\vert \psi_{ABC}\rangle = \frac{1}{\sqrt{2}} (\vert 000\rangle + \vert 111\rangle),
	\end{equation}
	written in the $Z$-basis.
\end{exmp}

\subsection{The proof of Theorem~\ref{Thm:NE_stabilizer}} \label{Sec:Proof}
The following three lemmas directly lead to the proof of Theorem~\ref{Thm:NE_stabilizer}.

\begin{lemma}
	$\Rk(J)$ is invariant under the change of the generators of $\mathcal{G}_{AB}$.
\end{lemma}
\begin{proof}
	Any change of generators can be done by a sequence of (1) permutations of generators (2) multiplying one generator to another. These operations induce an operation on the matrix $J$:
	\begin{enumerate}
		\item Switch two rows $a \leftrightarrow b$ and then switch two columns $a \leftrightarrow b$. (Here, $a\ne b$.)
		\item Add the $a$-th row to the $b$-th row and then add the $a$-th column to the $b$-th column over the field $F_2$. (Here, $a\ne b$.)
	\end{enumerate}
	It is easy to see that these two operations do not change $\Rk(J)$. This completes the proof.
\end{proof}

\begin{lemma}
	It is possible to choose a ``standard basis" of stabilizer generators of $\mathcal{G}_{AB}$:
	\begin{equation}
	\cup_{i=1}^{m_a}\{ a^i_{AB}, b^i_{AB}\}\cup_{s=1}^{m_c} \{ c^s_{AB}\}, \label{eq:standard_basis}
	\end{equation}	
	such that 
	\begin{enumerate}
		\item The nonnegative integers $m_a$ and $m_c$ satisfy
		\begin{equation}
		2 m_a + m_c = m_{AB}.
		\end{equation}
		\item When we write
		\begin{equation}
		\begin{aligned}
		a^i_{AB} &= a^i_A\otimes a^i_B ,\\
		b^j_{AB} &= b^j_A\otimes b^j_B, \\
		c^s_{AB} &= c^s_A\otimes c^s_B,
		\end{aligned}
		\end{equation}
		we have
		\begin{equation}\begin{aligned}
		\,	\{ a_A^i, b_A^i \} = 0, & \quad \forall i,\\
		\, [a_A^i, b_A^j] = 0, & \quad i\ne j,	\\
		\,	[a_A^i, a_A^{j}]= [b_A^i, b_A^{j}] = 0, & \quad \forall i,j,\\
		\,	[a_A^{i},c_A^s]=[b_A^{i},c_A^s] =0, & \quad \forall s,i,\\
		\, [c_A^s, c_A^t]=0, & \quad \forall s,t.
		\end{aligned}
		\end{equation}
		Equivalently, in the standard basis (\ref{eq:standard_basis}), we have a block-diagonal $J$:
		\begin{equation}
		J=\left( \begin{array}{cc}
		0&1\\
		1&0
		\end{array} \right)^{\oplus m_a} \oplus \left( \begin{array}{c}
		0
		\end{array} \right)^{\oplus m_c}.  \label{eq:starndard}
		\end{equation}
	\end{enumerate}
\end{lemma}
\begin{proof}
	Note that the matrix $J$, defined in (\ref{eq:def_J}) can also be treated as a skew-symmetric matrix over $F_2$. (On $F_2$, we have $1+1=0$.) We can apply the standard method to bring it to the standard form (\ref{eq:starndard}). The procedure is to applying a sequence of pairs of row and column operations, where the column operation is similar to the row operation. This sequence of operations is suitable for our purpose because it corresponds to a sequence of redefinition of stabilizer generators. This completes the proof.
\end{proof}

\begin{lemma}
	In terms of the number $m_a$ defined above,
	\begin{equation}
	E_N^{A\vert B}(\rho_{AB}) = m_a. \label{eq:NE_standard}
	\end{equation}
\end{lemma}
\begin{proof}
	The proof is based on a few simple observations:
	\begin{enumerate}
		\item Let $\{ h_{AB}^i \}_{i=1}^{m_{AB}}$ be the set of generators of stabilizer group $\mathcal{G}_{AB}$. Let 
		\begin{eqnarray}
		\tilde{h}_{AB}^i \equiv (h_{AB}^i)^{T_A}.
		\end{eqnarray} 
		Then $\{ \tilde{h}_{AB}^i \}_{i=1}^{m_{AB}}$ generates another stabilizer group, which we may denote as $\tilde{\mathcal{G}}_{AB}$.
		\item A simple property of partial transpose:
		\begin{equation}
		[(\lambda_A \otimes \lambda_B) \cdot (\mu_A\otimes \mu_B)]^{T_A} = (\mu_A^{T_A} \cdot \lambda_A^{T_A}) \otimes (\lambda_B \cdot \mu_B).
		\end{equation}
		\item We must have 
		\begin{equation}
		(h^i_{AB} h^j_{AB})^{T_A} = \pm \tilde{h}^i_{AB} \tilde{h}^j_{AB}
		\end{equation}
		because that each stabilizer generator is a tensor product of factors acting on $A$ and $B$.
		We obtain ``$+$" if $h^i_{A}$ and $h^j_A$ commute and we obtains ``$-$" if $h^i_{A}$ and $h^j_A$ anti-commute.
	\end{enumerate} 
	With these simple observations, we find that the density matrix $\rho_{AB}$, written the standard basis (\ref{eq:standard_basis}) as
	\begin{equation}
	\rho_{AB}= \frac{1}{2^{L_{AB}}} \prod_{i=1}^{m_a} (1 + a^i + b^i + a^ib^i)  \prod_{s=1}^{m_c} (1 + c^s) ,
	\end{equation}
	is of the following form after the partial transpose: 
	\begin{equation}
	\rho_{AB}^{T_A}= \frac{1}{2^{L_{AB}}} \prod_{i=1}^{m_a} (1 + \tilde{a}^i + \tilde{b}^i - \tilde{a}^i \tilde{b}^i)  \prod_{s=1}^{m_c} (1 + \tilde{c}^s) .
	\end{equation}
	Because different stabilizer generators $\{ \tilde{a}^i, \tilde{b}^j, \tilde{c}^s \}$ can independently take $\pm 1$, one can easily verify (\ref{eq:NE_standard}). This completes the proof.
\end{proof}

\section{Data collapse details}\label{appendix:data_collaps_details}

	Near the critical point $(p\approx p_c)$, we expect the physical quantities of interest, e.g., entanglement entropy, mutual information and entanglement negativity to scale as a function of $(p-p_c)L^{1/\nu}$. Here the dimensionless number $\nu$ as related to the correlation length by
	\begin{equation}
	\xi \sim \vert p-p_c \vert^{-\nu}.
	\end{equation}  
	The function can be different for different quantities. 
	
	To calculate $p_c$ and $\nu$, we perform data collapse for the mutual information $I_{A,B}$ and the entanglement negativity $E_N^{A\vert B}$ in the random Clifford circuit model with measurements:
	\begin{eqnarray}
	I_{A,B} = f((p-p_c)L^{1/\nu})\\
	E_N^{A\vert B}= \widetilde{f} ((p-p_c)L^{1/\nu}),
	\end{eqnarray}
	where we have fixed the two intervals to be antipodal regions with length $\vert A \vert = \vert B \vert = L/8$ whereas $f(\cdot)$ and $\widetilde{f}(\cdot)$ are functions. 
	
	We further perform data collapse for random Haar circuit model with measurements. This time, we consider the difference of half-chain entanglement entropy:
	\begin{equation}
	S_1(p)-S_1(p_c) = g((p-p_c)L^{1/\nu}),
	\end{equation}
	where $g(\cdot)$ is a function.

    By following the protocol in Ref.~\cite{datacollapse}, we find that $p_c = 0.16 (0.26), \nu=1.07(1.35)$ in the Clifford (Haar) case. The results for $p_c$ in both models are consistent with Ref.~\cite{2019PhRvB.100m4306L} and \cite{2019PhRvX...9c1009S}. For the critical exponent $\nu$, our result in the Clifford circuit model is close to that of Ref.~\cite{2019PhRvB.100m4306L}, but our results of Haar random circuit model differs with $\nu=2.01$ in Ref.~\cite{2019PhRvX...9c1009S}. This is understandable, as in the random Haar case the finite size effect is more significant.

    	\begin{figure}[h]
    	\includegraphics[width=0.7\columnwidth,valign=center]{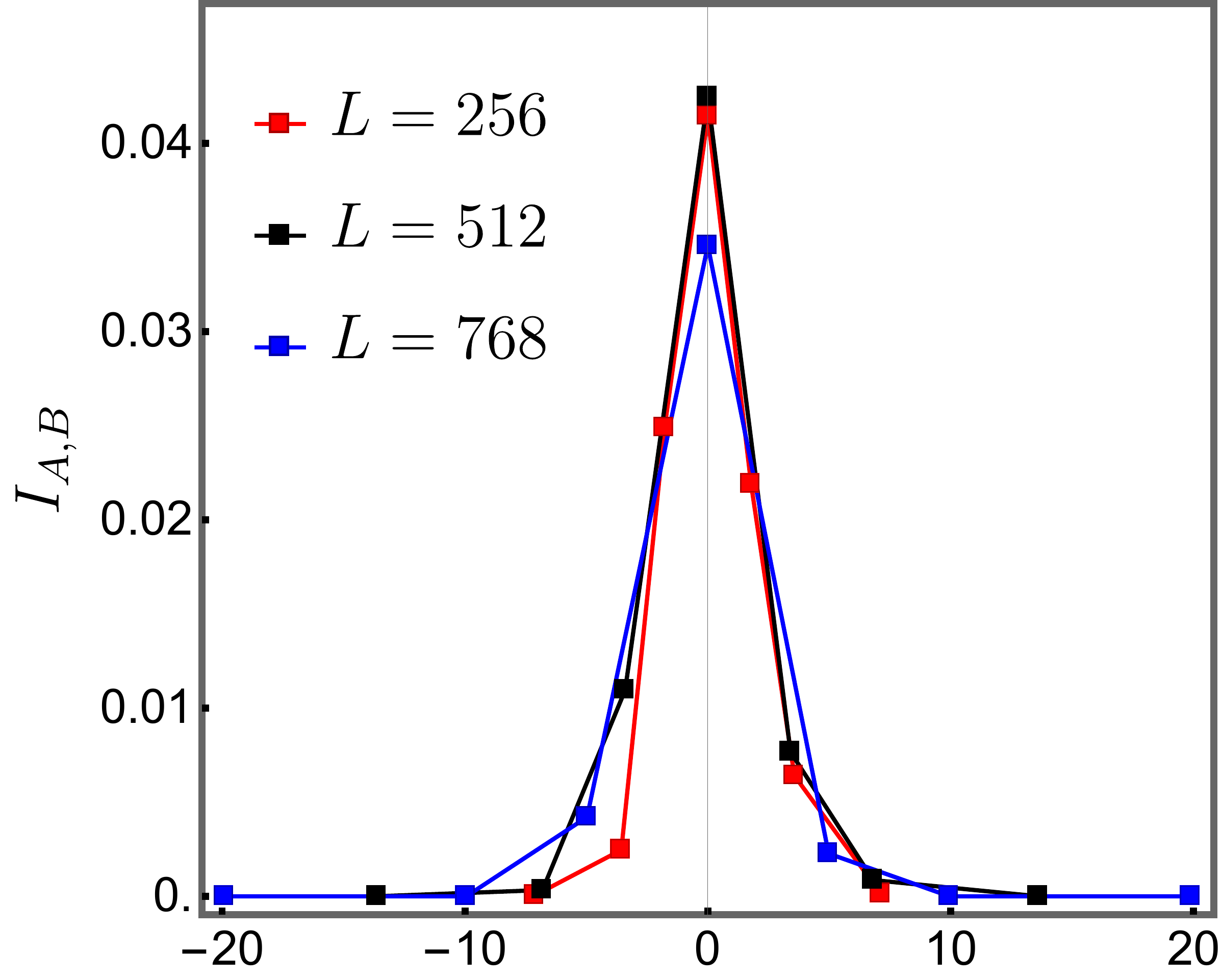}
    	\includegraphics[width=0.72\columnwidth,valign=center]{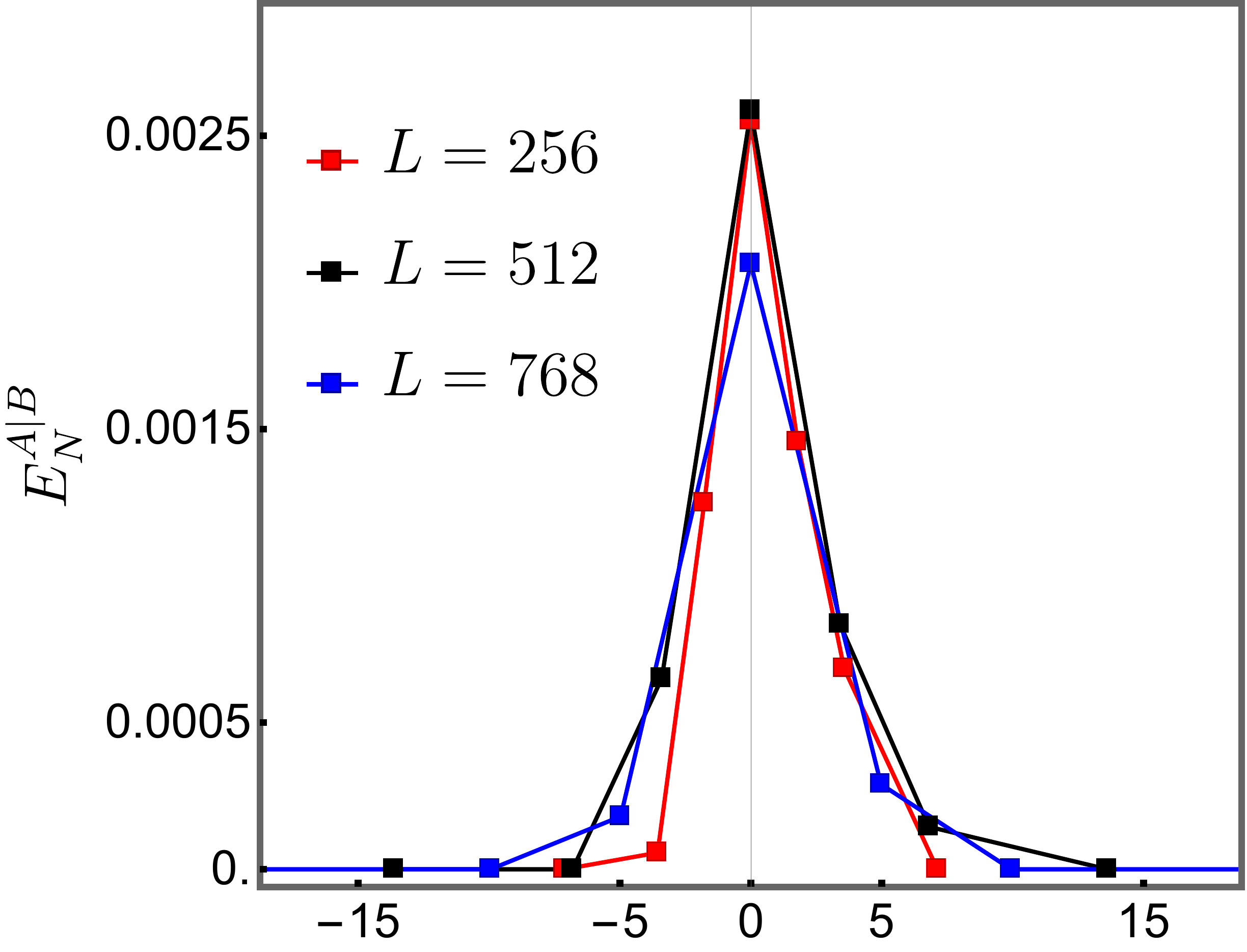}
    	\includegraphics[width=0.7\columnwidth,valign=center]{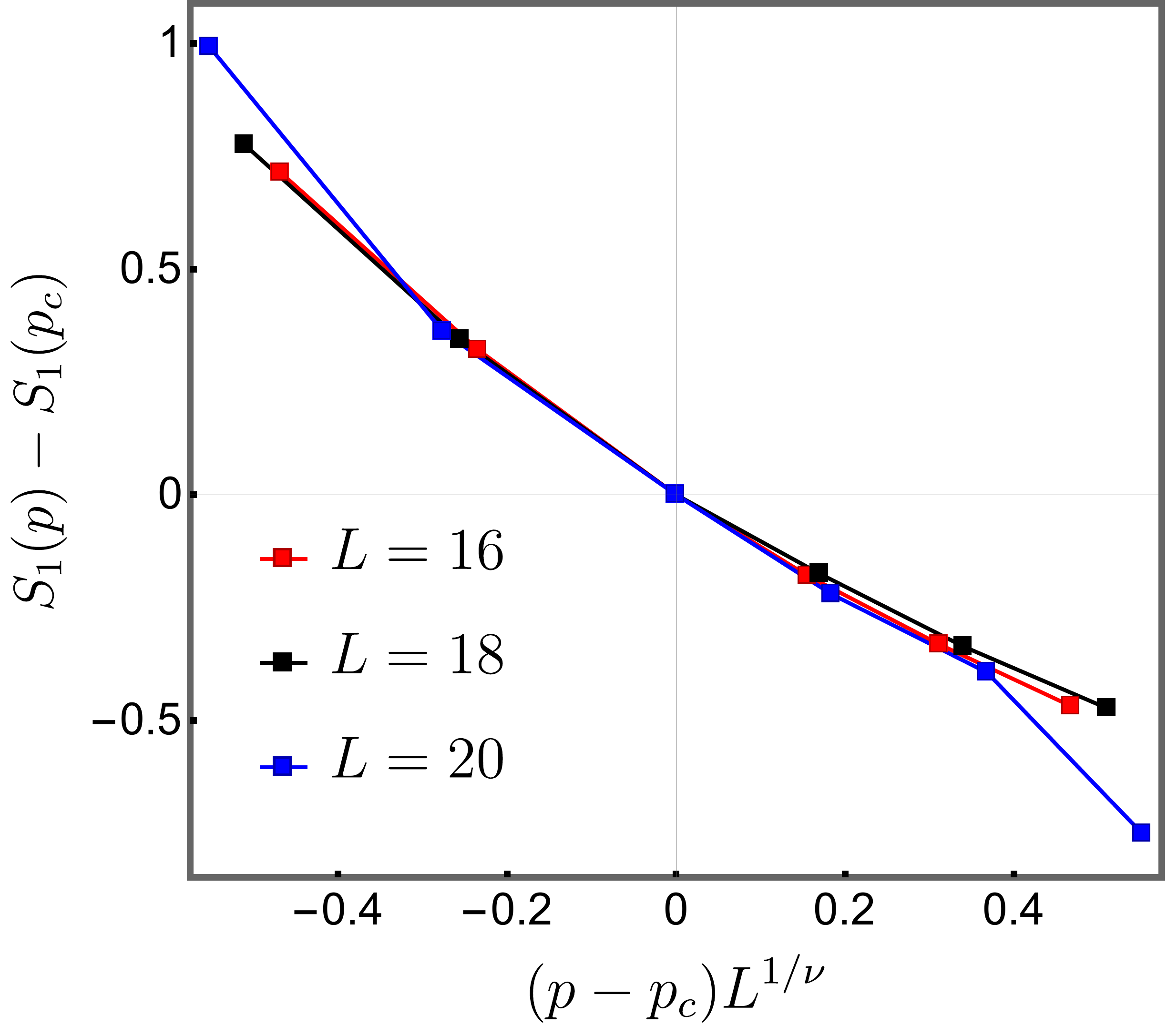}
    	\caption{(Top and middle) Data collapse for the mutual information $I_{A,B}$ and logarithmic negativity $E_N^{A\vert B}$ at the critical point $p_c=0.16$ for the hybrid Clifford circuit model.  Here we fix the two intervals to be antipodal regions with length $\vert A \vert = \vert B \vert = L/8$. Thus, the cross-ratio is fixed at $\eta= \sin^2(\frac{\pi}{8})\approx 0.146$.
    		(Bottom) Data collapse for the half-chain entanglement entropy $S_1(p)-S_1(p_c)$ for the random Haar model.}
    	\label{fig_data_collaps}
    \end{figure}
    \FloatBarrier

	\bibliography{ref}
	\bibliographystyle{apsrev}
\end{document}